\documentclass[12pt]{article}
\usepackage{cite}
\usepackage[a4paper,hscale=0.78,vscale=0.73]{geometry}
\usepackage{graphicx}
\usepackage{amssymb}
\usepackage{amsmath}
\usepackage{amsthm}
\usepackage{amscd}
\usepackage[all,cmtip]{xy} 
\numberwithin{equation}{section}
\usepackage{empheq}

\usepackage[british]{babel}

\author{Ziyang Hu\footnote{\texttt{z.hu@damtp.cam.ac.uk}}\\
D.A.M.T.P.\\
 University of Cambridge}
\title{On the General Problem of\\
Structure-Preserving Submersions}

\newcommand{\pd}{\partial}
\newcommand{\rs}{\mathbb{R}}
\newcommand{\cxs}{\mathbb{C}}
\DeclareMathOperator{\Ad}{Ad}
\DeclareMathOperator{\ad}{ad}
\DeclareMathOperator{\Ricci}{Ricci}

\newtheoremstyle{shape0}
  {9pt}
  {9pt}
  {}
  {}
  {\bfseries}
  {.}
  {.5em}
  {}

\newtheoremstyle{shape1}
  {9pt}
  {9pt}
  {\it}
  {}
  {\bfseries}
  {.}
  {.5em}
  {}

\newtheoremstyle{shape2}
  {9pt}
  {9pt}
  {}
  {}
  {\bfseries}
  {.}
  {.5em}
  {}

\theoremstyle{shape1}
\newtheorem{thm}{Theorem}
\newtheorem{prop}[thm]{Proposition}
\newtheorem{cor}[thm]{Corollary}

\theoremstyle{shape0}

\theoremstyle{shape2}
\newtheorem{dfn}[thm]{Definition}
\newtheorem{ex}[thm]{Example}
\theoremstyle{definition}

\begin{document}
\maketitle
\abstract{In this paper we give a general geometrical framework for working with problems that can be described as a structure-preserving submersion defined on a suitable space with a geometrical structure. We give many examples of how to formulate familiar problems arsing from physics in our framework.  Then, as an application of this framework we derive the new results of generalisations of the classical Herglotz-Noether theorem to arbitrary dimensions and to all conformally-flat ambient spaces. Possible further directions for research are also discussed.}
\newpage
\tableofcontents
\newpage
\section{Introduction}
\label{sec:introduction}

In physics and other fields where we use a geometrical space as a model for doing computations, situations that can be roughly described as ``structure-preserving submersions'' arise frequently. The archetypal example is given by the problem of rigid motion in classical mechanics where the geometrical structure of the moving body is preserved by the motion. In \cite{me:000} we see that a straightforward generalisation to special relativity is possible. Many techniques have been devised to solve these problems. However, these techniques are adapted to the particular problem at hand and, for example, techniques useful in the study of classical mechanical rigidity would not be very useful in the study of relativistic rigidity.

In this paper we describe a general geometrical framework encompassing all of the geometrical problems that can be described as a ``structure-preserving submersion''. This programme shall enable us to study many properties common to all such problems in an abstract way, while at the same time understand more deeply where many of the differences that do occur in the different problems arise. We shall also see that this framework gives a very robust foundation for doing computations, and problems that were computationally too expensive to carry out before is now much simpler. Most importantly, this framework gives us a very clear \emph{geometrical} picture of what is happening in a problem, and hence can effectively guide us through the computations without getting lost in the details.

Our plan for this paper is as follows. In section \ref{sec:structures-manifold} we will first describe how to formulate the concept of a ``manifold with structure'' within the Cartan framework, which covers all classical geometries and their inhomogeneous generalisations, for example, Riemannian geometry. In section \ref{sec:subm-immers-topol} we use the framework to first discuss how structure-preserving immersion fits into our programme, then goes into structure-preserving submersion proper. The reason for dealing with immersion first is that, this problem is simpler but shares some structural similarities with the submersion case, so that we can gain some experience, and more importantly, as we shall see, every submersion problem contains a family of submersion problems. We then give many simple examples of how familiar problems can be formulated in our framework, especially the examples of Riemannian and conformal submersions. Next, in section \ref{sec:riemannian-geometry}, we use our framework to derive some very rigid constraints on the forms of rigid flows in Riemannian spaces (and with a few change of signs, pseudo-Riemannian spacetimes). Our results will include a form of generalisation of the Herglotz-Noether theorem, namely that a rotational rigid motion in conformally flat space must be isometric in the Riemannian case, and a rotational conformally rigid motion in flat space must be conformally isometric in the conformal geometry case. This is a further extension to the result derived in \cite{me:000}, where more basic results are obtained using more down-to-earth methods but with much less theoretical underpinning. This will show the utility of our approach of solving problems. Along the way we will discuss several further directions worth pursuing, which we will summarise in section \ref{sec:conclusion}.

\section{Structures on a manifold}
\label{sec:structures-manifold}

In order to define a \emph{structure-preserving submersion} on a manifold we first need to understand what it means for a manifold to have a geometrical structure. A smooth manifold can be locally represented by $\rs^{n}$, which by itself is assumed to have no structure at all, and the simplest way to give it a structure is to specify the usual Euclidean metric on it. A general way of specifying the geometrical structure is to consider the transformation group preserving this geometrical structure instead: this is Klein's Erlangen programme. However, Klein's programme considers only \emph{homogeneous} geometrical structures, whereas in our application this condition is obviously too restrictive, and we need the full power of Cartan's extension of \emph{espaces g\'en\'eralis\'es}, which we shall now describe.

\subsection{Cartan geometry}
\label{sec:cartan-geometry}

Let $M$ be a manifold with a geometrical structure $\Gamma$ defined on it. A (local) \emph{admissible transformation} $f:M\rightarrow M$ is a (local) automorphism such that $f^{*}\Gamma = \Gamma$. We assume that the set of all admissible transformations $f$ locally have the structure of a Lie group $G$ --- this is not a very restrictive assumption at all. We assume further that $G$ acts transitively (otherwise we will simply use a quotient group) on $M$ on the right, and for any point $p\in M$, the isotropy group $H_{x}$ such that $f(x)=x$ for all $f\in H_{x}$ are all isomorphic, so we write $H_{x}\cong H$. Then the geometrical property $\Gamma$ is encoded in a very particular coframing on the principal $H$-bundle over $M$. First recall that a principal $H$ bundle is a smooth fibre bundle $\xi = (P,M,\pi,H)$ together with a right action $P\times H\rightarrow P$ that is fibre preserving and acts simply transitively on each fibre. Note that the fibre is isomorphic to the group, but which point in a fibre to choose as the identity element is arbitrary. This arbitrariness is captured by the right action. We are now ready to give the definition of a \emph{Cartan geometry}:
\begin{dfn}
  Let $G$ acts locally transitively on $M$ with isotropy group $H$. Let $\mathfrak{g}$ and $\mathfrak{h}$ be their respective Lie algebras. Then the Cartan geometry $\xi=(P,\omega)$ on $M$ modelled on $(\mathfrak g, \mathfrak h)$ consists of the following data:
  \begin{itemize}
  \item the smooth manifold M;
  \item the principal right $H$ bundle $P$ over $M$;
  \item the $\mathfrak g$-valued 1-form $\omega$ on $P$, called the \emph{Cartan connection}, satisfying
    \begin{enumerate}
    \item for each point $p\in P$, the linear map $\omega_{p}:T_{p}(P)\rightarrow\mathfrak g$ is an isomorphism;
    \item $(R_{h})^{*}\omega=\Ad(h^{-1})\omega$ for all $h\in H$ where $R$ denotes the right action of a group on itself and $\Ad$ denotes the adjoint action;
    \item $\omega(X^{\dagger})=X$ for all $X\in\mathfrak h$ where $X^{\dagger}$ is the vector field on $P$ generated by the infinitesimal group action $X\in\mathfrak g$.
    \end{enumerate}
  \end{itemize}
The $\mathfrak g$-valued 2-form $\Omega=d\omega+\tfrac{1}{2}[\omega\wedge\omega]$ is called the \emph{curvature}. If $\Omega$ takes value in the subalgebra $\mathfrak h$, then the geometry is called \emph{torsion free}. if $\Omega=0$, then the geometry is \emph{flat}.
\end{dfn}

When a geometry is flat, then $P\cong G$ and $\omega$ is the left-invariant Maurer-Cartan form on $G$. Even in the non-flat case, condition 3 above shows that restricted to each fibre $\omega$ is still the Maurer-Cartan form on $H$. The curvature measures the extent of the deformation of $P$ from the Lie group $G$. In the flat case, $M\cong G/H$ is a homogeneous space.

We will often need the notion of basic and semibasic forms. Let $H\rightarrow E\xrightarrow{\pi}M$ be a fibre bundle over $M$ with group action $H$. Then for the space of forms we have $\pi^{*}:\Lambda^{p}(M)\hookrightarrow\Lambda^{p}(E)$. For $\omega\in\Lambda^{p}(E)$, $\omega$ is \emph{basic} if it lies in the image of $\pi^{*}$, whereas it is \emph{semi-basic} if $\omega(V_{1},\dots,V_{n})=0$ whenever $V_{1}$ is tangential to a fibre. An easy criteria for a form to be basic is that it is semi-basic and right $H$-invariant.

\subsection{Examples: Riemannian geometry and conformal geometry}
\label{sec:exampl-riem-geom}

We now have all the abstract machinery we will need from Cartan's geometry. Next we will give two key examples of how to apply this programme. These examples will be used later when we construct immersions and submersions on it.

\begin{ex}
  [Riemannian geometry in the Cartan framework]
Let us see how we can capture the information contained in a Riemannian metric in the Cartan framework. We will take the model to be
\begin{align}
  \label{eq:1}
  G&=\left\{
    \begin{pmatrix}
      1&0\\
      V&A
    \end{pmatrix}\in M_{n+1}(\rs)\mid A\in O(n)
\right\}\\
  H&=\left\{
    \begin{pmatrix}
      1&0\\
      0&A
    \end{pmatrix}\in M_{n+1}(\rs)\mid A\in O(n)
\right\}\\
  \mathfrak g &=\left\{
    \begin{pmatrix}
      0&0\\
      v&a
    \end{pmatrix}\in M_{n+1}(\rs)\mid A\in \mathfrak{o}(n)
\right\}\\
  \mathfrak h &=\left\{
    \begin{pmatrix}
      0&0\\
      0&a
    \end{pmatrix}\in M_{n+1}(\rs)\mid A\in \mathfrak{o}(n)
\right\}
\end{align}
i.e., for the Lie algebra part, $a+a^{t}=0$. Using the form of the Lie algebra, the Cartan connection and curvature are
\begin{align}
  \label{eq:3}
    \omega&=
  \begin{pmatrix}
    0&0\\
    \theta&\alpha
  \end{pmatrix}\\
  \Omega&=
  \begin{pmatrix}
    0&0\\
    d\theta+\alpha\wedge\theta&d\alpha+\alpha\wedge\alpha
  \end{pmatrix}
\end{align}
The forms $\theta$ are basic, whereas the forms $\alpha$ are tangential to the fibre. We also see that the torsion free condition is just the usual ``Cartan's first structural equations'':
\begin{equation}
  \label{eq:4}
  d\theta=-\alpha\wedge\theta
\end{equation}
or, if we choose a suitable basis for $\mathfrak g$ and put on some indices (Einstein's summation convention applies),
\begin{equation}
  \label{eq:5}
  d\theta^{i}=-\alpha^{i}{}_{j}\wedge\theta^{j}.
\end{equation}
Therefore the data captured by a Riemannian metric is also captured by the 1-forms $\theta^{i}$ and $\alpha^{i}{}_{j}$ satisfying the torsion-free condition \eqref{eq:5}. For flat space (Euclidean space), we of course also have
\begin{equation}
  \label{eq:6}
  d\alpha^{i}{}_{j}=-\alpha^{i}{}_{k}\wedge\alpha^{k}{}_{j}.
\end{equation}

\end{ex}
\begin{ex}[Conformal geometry in the Cartan framework]
A conformal structure on a manifold $M$ is an equivalence class of Riemannian structure on $M$: two Riemannian metrics $g$ and $g'$ belongs to the same class if and only if there is a function $\Lambda:M\rightarrow \rs$ such that $g'=\Lambda g$. We need to encode this information into a linear Lie group acting on the tangent space, and it is fairly complicated, as shown below.

We will use the M\"obius model: let $\mathbb{L}=\rs^{n+2}$ be equipped with the indefinite metric
\begin{equation}
  \label{eq:42}
  \Sigma_{n+1,1}=
  \begin{pmatrix}
    0&0&-1\\
    0&I_{n}&0\\
    -1&0&0
  \end{pmatrix}
\end{equation}
 this is the so-called light-cone version of the metric. Without regard of the metric, we can form the projective space $\mathbb{P}(\mathbb{L})$ of $\mathbb{L}$, i.e., two points are identified if and only if they lie on the same straight line through the origin. Using the metric now, we define the \emph{M\"obius $n$-space} $M_{0}$ to be the set of lightlike points in $\mathbb{P}(\mathbb{L})$. Note that the affine map $\rs^{n+1}\rightarrow\mathbb{L}$
\begin{equation}
  \label{eq:43}
  (y_{0},y_{1},\dots,y_{n})\mapsto\left(\frac{1}{\sqrt{2}}(1+y_{0}),y_{1},\dots,y_{n},\frac{1}{\sqrt{2}}(1-y_{0})\right)
\end{equation}
induces a diffeomorphism between the $n$-sphere and the M\"obius $n$-space, hence also the terminology \emph{M\"obius $n$-sphere}. 

The Lorentz group is the symmetry group of the metric $\Sigma_{n+1,1}$ fixing the origin, and the symmetry group of the M\"obius $n$-sphere is a subgroup.  The kernel of the group action of Lorentz group on $M_{0}$ is $\{\pm I\}$. For the isotropy subgroup, we will take the one fixing the north pole of the M\"obius sphere. Hence we will take as our model the \emph{M\"obius model}
\begin{align}
  \label{eq:44}
  G&=O(n+1,1)/\{\pm 1\}\\
  H&=\{h\in G\mid h[e_{0}]=[e_{0}]\},\qquad\text{$[e_{0}]$ refers to the north pole.}
\end{align}
The  matrix representation of $H$ is given by
\begin{equation}
  \label{eq:45}
  H=\left\{
    \begin{pmatrix}
      z&0&0\\
      0&a&0\\
      0&0&z^{-1}
    \end{pmatrix}
    \begin{pmatrix}
      1&q&r\\
      0&1&q^{t}\\
      0&0&1
    \end{pmatrix}
    \mid
    z\in\rs^{+},a\in O(n),q^{t}\in\rs^{n},r=\frac{1}{2}qq^{t}.
\right\}
\end{equation}
and for the Lie algebra pair $(\mathfrak g,\mathfrak h)$, we will take
\begin{align}
  \label{eq:46}
  \mathfrak g&=\left\{
  \begin{pmatrix}
    z&q&0\\
    p&s&q^{t}\\
    0&p^{t}&-z
  \end{pmatrix}
  \mid s^{t}=s\right\}
\\
  \mathfrak h&=
  \left\{
  \begin{pmatrix}
    z&q&0\\
    0&s&q^{t}\\
    0&0&-z
  \end{pmatrix}
  \mid s^{t}=s\right\}
\end{align}
Note that the adjoint action of $h\in H$ on $p\in\mathfrak g/\mathfrak h$ is
\begin{equation}
  \label{eq:52}
  \Ad(h)p=sz^{-1}p
\end{equation}
as expected, i.e., rotation and scaling.

With $(\mathfrak g,\mathfrak h)$ and $H$, we are now able to define a Cartan geometry modelled on M\"obius $n$-sphere. We will denote the Cartan connection by
\begin{equation}
  \label{eq:47}
  \omega=
  \begin{pmatrix}
    \epsilon&\nu&0\\
    \theta&\alpha&\nu^{t}\\
    0&\theta^{t}&-\epsilon
  \end{pmatrix}
\end{equation}
note that the forms $\theta$ plays the same role as the Riemannian $\theta$: they are basic for the canonical projection of the bundle. The curvature is
\begin{align}
  \label{eq:49}
  \Omega&=
  \begin{pmatrix}
    E&Y&0\\
    \Theta&A&Y^{t}\\
    0&-\Theta^{t}&-E
  \end{pmatrix}\\
  &=
  \begin{pmatrix}
    d\epsilon+\nu\wedge\theta&d\nu+\epsilon\wedge\nu+\nu\wedge\alpha&0\\
    d\theta+\theta\wedge\epsilon+\alpha\wedge\theta&d\alpha+\theta\wedge\nu+\alpha\wedge\alpha+\nu^{t}\wedge\theta^{t}&\star\\
    0&\star&\star
  \end{pmatrix}
\end{align}
The $A$ block of the curvature plays the same role as the torsion-free part of the Riemannian curvature. Call this part of the Lie algebra the $\mathfrak{s}$-block. We also define the Ricci homomorphism, the abstract version of the Ricci contraction of the Riemannian tensor,  in the present context defined by
\begin{align}
  \label{eq:50}
  \Ricci:\hom(\Lambda^{2}(\mathfrak g/\mathfrak h),\mathfrak s)&\cong\Lambda^{2}(\mathfrak g/\mathfrak h)^{*}\otimes\mathfrak s\\
  &\xrightarrow{id\otimes\ad}\Lambda^{2}(\mathfrak g/\mathfrak h)^{*}\otimes\mathfrak g/\mathfrak h\otimes(\mathfrak g\otimes\mathfrak h)^{*}\\
  &\xrightarrow{\text{contract}\otimes\ad}(\mathfrak g/\mathfrak h)^{*}\otimes(\mathfrak g/\mathfrak h)^{*}
\end{align}
In the above, $\ad$ is the adjoint action of $\mathfrak s$ on $\mathfrak g/\mathfrak h$. This is the usual interpretation of $\mathfrak s$ as a matrix of transformations. The last step is the map
\begin{equation}
  \label{eq:51}
  t^{*}\wedge u^{*}\otimes v\otimes w^{*}\mapsto (u^{*}(v)t^{*}-t^{*}(v)u^{*})\otimes w^{*}
\end{equation}
 If the curvature blocks $E=0$ (no scaling curvature), $\Theta=0$ (no torsion), and $A$ is in the kernel of $\Ricci$, then the M\"obius geometry is called \emph{normal}. \emph{Normal M\"obius geometries are in 1-1 correspondence with conformal metrics on a manifold}. This completes the formulation of conformal geometry in the Cartan language.

\end{ex}

Let us remark that the term \emph{Cartan geometry} does not stand for a single, concrete model of geometry, as in the case of Euclidean geometry. Rather, it is an abstract geometrical language for describing other geometries. In particular, we should note that though in a concrete problems the Cartan connection $\omega$ can be given by an explicit expression, it is more useful in the Cartan framework to first study the properties of it only abstractly, without any particular expressions attached, just as we can study the properties of a Riemannian metric in Riemannian geometry without giving it any particular expression. The Cartan connection is the central concept for Cartan geometry, as it encodes all data about the geometry.

\subsection{Lifting and the method of equivalence}
\label{sec:lift-into-princ}
In fact, the theory of Cartan geometry arises as an application of Cartan's algorithm for solving equivalence problems. However, the method of equivalence is more general for solving a large class of problems in differential systems. The method of equivalence is very powerful due to its algorithmic nature, and once a problem can be formulated in its form a solution is sure to be found, though the computation may be very complicated so that the use of computer is required.

We shall now briefly outline the method of equivalence so as to make a more computational interpretation of our later manipulations possible. We will gloss over a large number of technical issues, as for our need  just  keeping a simple picture in mind is sufficient, especially when we discuss why the structure-preserving submersion problem \emph{cannot} be solved in section \ref{sec:equiv-meth-fails}. For further details of the method, see \cite{LBryant:1991p8950,Fels:1998p7335,Fels:1999p7361}.

Basically, the method of equivalence is concerned with the following question: given two manifolds $M$ and $N$, each with a set of differential 1-forms $\theta_{M}^{i}$ and $\theta_{N}^{i}$ on it, does that exist a diffeomorphism $\Phi:M\rightarrow N$ such that
\begin{equation}
  \label{eq:78}
  \Phi^{*}\theta_{N}^{i}=g^{i}{}_{j}\theta_{M}^{j},
\end{equation}
 where $g^{i}{}_{j}\in G$ is a specified Lie group acting linearly on the column vector of 1-forms? It is neither assumed that the 1-forms are linearly independent nor that they span the cotangent space, though in the case where they do not span the cotangent space we need to complement them with other forms. The 1-forms encode the geometrical structure in the problem with respect to the cotangent space, and the Lie group $G$ is simply the group that preserves this structure. For example, for Riemannian geometry \emph{an} orthonormal frame encodes the same amount of information as the metric: we can write $g=\sum\theta^{i}\otimes\theta^{i}$ for some one forms $\theta^{i}$, using the Gram-Schmidt algorithm. The group $G$ is then simply the orthogonal group acting on the cotangent space in the usual way, and formulated in this way, this problem is the same as finding whether two manifolds as ``the same'' within the realm of Riemannian geometry, the problem Riemann first studied. However, \eqref{eq:78} is not symmetric in $M$ and $N$, and it is cumbersome to have elements of $G$ explicitly around, so we use the following method: we ``lift'' the problem into the space $M\times G$ and $N\times G$, and define the lifted 1-forms by $\omega^{i}_{M}=g\cdot\pi_{M}^{*}\theta^{i}_{M}$ and $\omega^{i}_{N}=g\cdot\pi_{N}^{*}\theta^{i}_{N}$, where $\pi_{M}:M\times G\rightarrow M$ is the projection and similarly for $N$. It can then be proved that the original problem \eqref{eq:78} is satisfied if and only if there exists diffeomorphism $\tilde \Phi:M\times G\rightarrow N\times G$ such that $\tilde\Phi^{*}\omega_{N}^{i}=\omega_{M}^{i}$. Now we can see how the principal bundle in Cartan geometry arises in this way. The case of Cartan connection arises also from this procedure: if the original 1-forms \emph{are} a coframe, then $\omega_{M}^{i}$ etc., already form a basis of basic 1-forms. The vertical 1-forms arise by considering the differentials:
 \begin{align}
   \label{eq:79}
   d\omega^{i}&=dg\wedge\theta^{i}+g\cdot d\theta^{i}\\
   &=dg\, g^{-1}\wedge g\cdot\theta^{i}+\gamma^{i}{}_{jk}\theta^{j}\wedge\theta^{k}\\
   &=dg\, g^{-1}\wedge \omega^{i}+\Gamma^{i}{}_{jk}\omega^{j}\wedge\omega^{k}
 \end{align}
now $dg\,g^{-1}$ is the Maurer-Cartan form in the fibre direction, thus we have a basis of 1-forms on the lifted manifold. The functions $\Gamma^{i}{}_{jk}$ are called \emph{torsion} and depends explicitly on the group elements. It is possible to absorb some or all of the torsion by redefining the  Maurer-Cartan forms by adding linear combinations of $\omega^{i}$ while still keeping them Lie algebra-valued. The aim is to make the torsion independent of the group elements by a series of such absorptions, and then the remaining essential torsions are just differential invariants of the problem, and their functional relationships determine whether the original two sets of 1-forms are equivalent. In the case of geometrical coframes, after absorption we have the Cartan connection.

\section{Structure-preserving maps}
\label{sec:subm-immers-topol}

We start by recalling the definition of immersion and submersion in differential \emph{topology}:

\begin{dfn}
  A smooth map $f\rightarrow M^{m}\rightarrow N^{n}$ with constant rank $r$ is called an \emph{immersion} if $r=m$, and a \emph{submersion} if $r=n$.
\end{dfn}

For an immersion or a submersion, an easy application of the implicit function theorem yields the following well-known-result:

\begin{prop}
  Let $f\rightarrow M^{m}\rightarrow N^{n}$ be an immersion (resp.~submersion). Then for each point $p\in M$ there are coordinate systems $(U,\varphi),(V,\psi)$ about $p$ and $f(p)$, respectively, that the composite $\psi f \varphi^{-1}$ is a restriction of the coordinate inclusion $\imath : \rs^{m}\rightarrow\rs^{m}\times\rs^{m-n}$ (resp.~a restriction of the coordinate projection $\pi: \rs^{n}\times \rs^{m-n}\rightarrow\rs^{n}$).
\end{prop}

This important property tells us two things. First, if we only care about local properties of an immersion or submersion, then the above proposition allows us to construct charts such that $f(M)$ in the case of immersion or $N$ in the case of a submersion is locally a topological manifold. Note that this may be true only locally: a map $f:\rs\rightarrow\rs^{2}$ given by $t\mapsto (t, at)$ with $a$ irrational, composed with the covering map $p:\rs\rightarrow T^{2}=\rs^{2}/\mathbb{Z}^{2}$ is an immersion, but $pf(\rs)$ is not globally a topological manifold since the image is dense in $T^{2}$. Nonetheless if we only focus on the local properties such subtleties will not bother us. The second thing this property tells us is the \emph{equivalence problem} of an immersion or submersion: given two local immersions (resp.~submersions) in a manifold, does there exist an \emph{admissible} transformation of the manifold such that the first immersion (resp.~submersion) can be transformed to coincide with the second? Here admissible transformations refer to all (local) diffeomorphisms of the manifold since we are working in the category of smooth manifolds. Since the property has given us coordinates charts in which all immersions or submersions look locally the same, the solution to the equivalence problem is

\begin{cor}
  In the category of smooth manifolds, all immersions (submersions) are locally equivalent, provided the dimensions of the relevant manifolds match.
\end{cor}

Note that the implicit function theorem is really \emph{implicit}: in general, there is no explicit algorithm telling us what the required transformation is. In our study that follows which deals with geometry rather than topology, we would want explicit algorithms whenever possible.

For many applications, the category of smooth manifolds is too general to be of any use. For example, we would certainly object to the statement that a sphere is locally the same as a plane in Euclidean space. The reason is that an Euclidean space has more structure (i.e., the Euclidean metric) to distinguish differences not seen by the smooth structure. By focusing on the extra structure, we have made the transition from the field of differential topology to differential geometry. In what follows, we shall investigate the equivalence problem of immersions and submersions in many differential geometrical settings. We will ask ourselves the following two questions:
\begin{itemize}
\item Given an immersion or submersion $f:M\rightarrow N$ where in the case of immersion $N$ has an extra geometrical structure (e.g.~a Riemannian metric) defined on it and in the case of submersion $M$ has an extra geometrical structure defined on it, when does it follow that $f(M)$ in the case of immersion or $N$ in the case of a submersion is also a manifold with extra geometrical structure?
\item In the case where we can define an extra geometrical structure on $f(M)$ or $N$, to solve the equivalence problem taking into consideration of the extra geometrical structures.
\end{itemize}
As we are not interested in any particular configuration but instead are concerned with the general problem, it is a tremendous advantage if we can work in an explicitly coordinate-independent way. Of course, this is now easy for us, as we already know how to apply Cartan's constructions to give geometrical structures to manifolds, as discussed in section \ref{sec:structures-manifold}.

\subsection{Structure-preserving immersions}
\label{sec:immers-subm}

Although this paper is mainly concerned with the problem of structure-preserving submersions, there are two reasons why we should discuss structure-preserving immersions here first. One reason is that submersion and immersion has many structural similarities, and when we study immersions we gain some experience of how to deal with certain quantities; a more important reason is that, as we will see, every submersion problem contains a large class of immersion problems.

To keep discussion short, we will just give two examples which should sufficiently illustrate the methods dealing with immersion problems. For the general theory, see \cite{Fels:1999p7361}.

\begin{ex}
[Riemannian immersion]
Suppose $N$ is equipped with a Riemannian geometry and we have an immersion $f:M\rightarrow N$. First we need to place our bundle in a normalised form. Using $R_{h}^{*}\omega=\Ad(h^{-1})\omega$, if we apply the action of
\begin{equation}
  \label{eq:7}
  h=
  \begin{pmatrix}
    1&0\\
    0&r
  \end{pmatrix}
\end{equation}
Then an easy calculation shows that $R_{h}^{*}\theta=r^{-1}\theta$, $r\in O(n)$. Therefore, we can use this action to align the first $m$ of the basic forms $\theta$ to be tangential to $f(M)$. Then, restricted to $f(M)$, the Cartan connection will be of the following form
\begin{equation}
  \label{eq:8}
  \omega=
  \begin{pmatrix}
    0&0&0\\
    \theta&\alpha&-\beta^{t}\\
    0&\beta&\gamma
  \end{pmatrix}
\end{equation}
where $\alpha$ is a skew $m\times m$ matrix of 1-forms and $\gamma$ is a skew $(n-m)\times(n-m)$ matrix of 1-forms. To keep the connection in this nice form, we are no longer allowed arbitrary right actions of $H$. Instead, now the allowed actions must be of the form
\begin{equation}
  \label{eq:9}
  h=
  \begin{pmatrix}
    1&0&0\\
    0&r&0\\
    0&0&s
  \end{pmatrix}
\end{equation}
where $r\in O(m)$, $s\in O(n)$. This is a \emph{reduction of the principal bundle $P$}. Let us still call the reduced bundle $P$. In the reduced bundle, $\alpha$ and $\gamma$ are still tangential to the fibre, but the tangential directions in which $\beta$ lives are now gone. Since $\beta$ derives from the Maurer-Cartan form on $H$, it now cannot have anything to do with $\alpha$ and $\gamma$. Since $\alpha$, $\beta$, $\gamma$ still constitute a basis of 1-forms on $P$, we must have $\beta$ a linear combination of the basic forms $\theta$. If we use lowercase Latin indices for the tangential directions and uppercase for the normal directions, we have
\begin{equation}
\label{eq:10}
\beta^{i}{}_{A}=M^{i}{}_{jA}\theta^{j}
\end{equation}
for \emph{functions} $M^{i}{}_{jA}$ on $P$. The torsion-free condition also has to be imposed. An easy calculation shows that this now amounts to the condition
\begin{equation}
  \label{eq:11}
  \beta\wedge\theta=0
\end{equation}
or, using \eqref{eq:10}
\begin{equation}
  \label{eq:12}
  M^{i}{}_{jA}\theta^{i}\wedge\theta^{j}=0
\end{equation}
from which we conclude that the functions $M^{i}{}_{jA}$ is symmetric in the tangential indices. To summarise, for an immersion $f:M\rightarrow N$ in Riemannian geometry, we have the following piece of data:
\begin{itemize}
\item the Cartan connection
  \begin{equation}
    \label{eq:13}
    \begin{pmatrix}
      0&0\\
      \theta&\alpha
    \end{pmatrix}
  \end{equation}
on $f(M)$, which determines the induced geometry on $M$ (and the induced Riemannian metric);
\item the Ehresmann connection $\gamma$ on $f(M)$, which determines the geometry that is normal to $f(M)$;
\item the gluing data $M^{i}{}_{jA}$ which ties together the tangential and normal part of $f(M)$.
\end{itemize}
We can of course dig deeper by studying interesting geometrical properties if the ambient space is Euclidean, and study the classification by asking the equivalence question. Since these are well-studied in the literature, we shall content with only sketching the above geometrical framework for beginning such investigations.

\end{ex}
\begin{ex}
  [Conformal immersion]
Let $N$ be a normal M\"obius geometry and $f:M\rightarrow N$ an immersion. As in the Riemannian case, the first step is using the adjoint action \eqref{eq:52} to reduce the principal bundle by aligning the first several $\theta$ with the submanifold and restrict to the submanifold. After this reduction, the Cartan connection becomes
\begin{equation}
  \label{eq:53}
  \omega=
  \begin{pmatrix}
    \epsilon&\nu&\mu&0\\
    \theta&\alpha&-\beta^{t}&\nu^{t}\\
    0&\beta&\gamma&\mu^{t}\\
    0&\theta^{t}&0&-\epsilon
  \end{pmatrix}
\end{equation}
the group $H$ has been reduced to
\begin{equation}
  \label{eq:54}
  H=\left\{
  \begin{pmatrix}
    1&p&q&s\\
    0&1&0&p^{t}\\
    0&0&1&q^{t}\\
    0&0&0&1
  \end{pmatrix}
  \begin{pmatrix}
    z&0&0&0\\
    0&a&0&0\\
    0&0&c&0\\
    0&0&0&z^{-1}
  \end{pmatrix}
\mid a\in O(m),c\in O(m-n),s=\frac{1}{2}(pp^{t}+qq^{t})
\right\}
\end{equation}
i.e., in terms of Lie algebra
\begin{equation}
  \label{eq:55}
  \begin{pmatrix}
    \star&\star&\star&0\\
    0&\star&0&\star\\
    0&0&\star&\star\\
    0&0&0&\star
  \end{pmatrix}
\end{equation}
As in the Riemannian case, after the reduction $\beta$ becomes a basic form. Hence we write
\begin{equation}
  \label{eq:56}
  \beta^{A}{}_{i}=M^{A}{}_{ij}\theta^{j}
\end{equation}
the torsion free condition again means that we need to ensure
\begin{equation}
  \label{eq:57}
  0=\beta\wedge\theta=M^{A}{}_{ij}\theta^{j}\wedge\theta^{i}
\end{equation}
i.e., $M^{A}{}_{ij}$ is symmetric in $i$ and $j$.

In the Riemannian case, we stopped here since those were all the \emph{structural} things we could do. But since M\"obius geometry has more freedom, now we can do still more. Taking $h$ as an element of $H$ of the form \eqref{eq:54}, we can verify that
\begin{equation}
  \label{eq:58}
  R_{h}^{*}M^{A}{}_{ij}=z^{-1}a^{-1}(c(M-qI))a=z^{-1}(a^{-1})^{i}{}_{k}(c^{A}{}_{B}(M^{B}{}_{kl}-q^{B}\delta_{kl}))a^{l}{}_{j}
\end{equation}
By choosing $q$, we can make $M$ trace-free: $\sum_{i}M^{A}{}_{ii}=0$. After this second reduction, the group becomes
\begin{equation}
  \label{eq:59}
    H=\left\{
  \begin{pmatrix}
    1&p&0&s\\
    0&1&0&p^{t}\\
    0&0&1&0\\
    0&0&0&1
  \end{pmatrix}
  \begin{pmatrix}
    z&0&0&0\\
    0&a&0&0\\
    0&0&c&0\\
    0&0&0&z^{-1}
  \end{pmatrix}
\mid a\in O(m),c\in O(m-n),s=\frac{1}{2}pp^{t}
\right\}
\end{equation}
i.e., in terms of Lie algebra,
\begin{equation}
  \label{eq:60}
  \begin{pmatrix}
    \star&\star&0&0\\
    0&\star&0&\star\\
    0&0&\star&0\\
    0&0&0&\star
  \end{pmatrix}
\end{equation}
The Cartan connection is still given by \eqref{eq:47}, but now in addition $\mu$ is semi-basic. We note here that the forms $\mu$ and $\nu$ are not very important since for $n\neq 2$ they are completely determined by the requirement of a normal geometry.
\end{ex}

Before studying submersions, let us remark here that immersions are much nicer than submersions. Given an immersion map and a geometrical structure on the range of the map, in most circumstances we can define a geometrical structure on the domain of the map for which the map is structure-preserving, and in a large number of cases this structure is even unique. Also, as in our last part of discussion of Riemannian immersions, it is possible to only consider the ``ambient geometry'', without reference to the ambient space at all and discard all information that is not relevant to the immersion.

\subsection{Anatomy of the structure-preserving submersion}
\label{sec:anat-struct-pres}

I hope the above examples of structure-preserving immersions should at least give a rough idea of how a structure-preserving submersion should look like. Now let us consider the following. Suppose $\sigma:M\rightarrow B$ is a submersion that can be described as ``structure-preserving''. Then \emph{both} $M$ and $B$ must have structure related to them. As we saw before, the structures are best considered in the principal bundle, hence we shall aim to find a map which covers $\sigma$ in the diagram below:
\begin{equation}
  \label{eq:40}
  \xymatrix{
M\times H\ar[r]^{\approx}\ar[rd]_{\text{proj}_{1}}&P\ar@{..>}[r]^{?}\ar[d]_{\pi_{M}}&Q\ar[d]^{\pi_{B}}&B\times H'\ar[l]_{\approx}\ar[ld]^{\text{proj}_{1}}\\
&M\ar[r]^{\sigma}&B&
}
\end{equation}
It is obvious that $H'$ needs to be a subgroup of $H$. Also, for this to be a covering map, the map we are looking for needs to be surjective. However, this implies there must be a surjective homomorphism from $H$ to its subgroup $H'$. In many applications where a structure-preserving submersion is obvious, the group $H$ is simple, and hence no such homomorphism can exist. Therefore we need to look for a more complicated setting.

Let us for the moment forget that we are dealing with \emph{structure-preserving} maps and focus only on the submersion. The submersion locally defines a foliation on the manifold $M$, and a subgroup $H_{L}\subset H$ can be defined that preserves this foliation. In other words, let $\theta^{a}$ and $\theta^{i}$ be a basis of 1-forms on the manifold which generates the Cartan connection in the principal bundle, and let $\theta^{i}$ define the Frobenius distribution defining the submersion. Then
\begin{equation}
  \label{eq:61}
  H_{L}=\{g\in H\mid g\cdot \omega^{i}\equiv0\mod \{\omega^{j}\},\forall i\}
\end{equation}
Geometrically, this amounts to a reduction of the principal bundle, as pictured below:
\begin{equation}
  \label{eq:48}
  \xymatrix{
M\times H\ar[r]^{\approx}\ar[rd]_{\text{proj}_{1}}&P\ar[d]_{\pi_{M}}&P'\ar[l]_{\rho}\ar[d]^{\pi'_{M}}&M\times H_{L}\ar[l]_{\approx}\ar[ld]^{\text{proj}_{1}}\\
&M&\ar[l]_{=}M&
}
\end{equation}
By considering its action on $\theta^{A}$, we can see that as a subgroup of $GL(n)$, the group $H_{L}$ has block upper-triangular structure:
\begin{equation}
  \label{eq:70}
  g\cdot
  \begin{pmatrix}
    \{\theta^{i}\}\\
    \{\theta^{A}\}
  \end{pmatrix}
=
\begin{pmatrix}
  g_{ij}&g_{iA}\\
  0&g_{AB}
\end{pmatrix}
\begin{pmatrix}
  \{\theta^{i}\}\\
  \{\theta^{A}\}
\end{pmatrix}.
\end{equation}
An important note: the above representation of the group should be understood in abstract terms, together with the labelling: as we saw in the case of conformal geometry, it is not necessarily true that the group action is always aligned in this way (though if we use bigger matrix and padding a sufficient number of zeros in the column vector, it is always possible to do so, and the way to do it should be clear from the problem at hand, as in the case of conformal geometry).

Since we must have a submersion before we can talk about structure-preservation, it makes sense that we take this as our starting point. Therefore, we look for the map indicated by $\tilde\sigma$ below:
\begin{equation}
  \label{eq:71}
  \xymatrix{
M\times H\ar[d]^{\approx}&\ar[l]_{\rho}M\times H_{L}\ar[d]^{\approx}\ar[r]^{\tilde\sigma}& B\times H'\ar[d]^{\approx}\\
P\ar[d]^{\pi_{M}}&\ar[l]_{\rho}P'\ar[d]^{\pi'_{M}}\ar[r]^{\tilde\sigma}&Q\ar[d]^{\pi_{B}}\\
M&\ar[l]_{=}M\ar[r]^{\sigma}&B
}
\end{equation}
In the above diagram, we require all squares commute, all arrows to the left injective, and all arrows to the right surjective. Note that both $P$ and $Q$ are Cartan-geometries in their own right, and hence we can use the above maps to pull-back their Cartan connections to $P'$:
\begin{equation}\xymatrix@R=3pt{
P&\ar[l]_{\rho}P'\ar[r]^{\tilde\sigma}&Q\\
\omega_{P}\ar@{|->}[r]&\rho^{*}\omega_{P}&\\
&\tilde\sigma^{*}\omega_{Q}&\ar@{|->}[l]\omega_{Q}.
}
\end{equation}
We now want to find a basis of 1-forms on $P'$. This should also give us the constraints on the map $\tilde\sigma$. The pulled-back forms on $P'$ are (we will omit some pullback signs when this should cause no confusion of on which space we are currently working):
\begin{equation}
  \label{eq:73}
  \begin{array}{rcccccc}
  \rho^{*}\omega_{P}:&\theta_{P}^{A}&\theta_{P}^{i}&\omega_{P}^{AB}&\omega_{P}^{ij}&\omega_{P}^{Ai}&\omega_{P}^{iA}\\
  \tilde\sigma^{*}\omega_{Q}:&&\theta_{Q}^{i}&&\omega_{Q}^{ij}&&
\end{array}
\end{equation}
It is obvious that the forms $\rho^{*}\omega_{P}$ span $P'$, whereas the forms $\tilde\sigma^{*}\omega_{Q}$ are linearly independent on $P$. Moreover, as both $\theta^{i}_{P}$ and $\theta^{i}_{Q}$ are basic with respect to projection to $M$, they span the same subspace of the tangent bundle of $P$. However, there is an important difference between them.

Suppose we choose a point $q\in Q$, and consider $\tilde L(q)=\tilde\sigma^{-1} (q)$. This set can be given a submanifold structure, of dimension $\dim H_{L}-\dim H'+\dim H-\dim B$. It can even be given a $H'$ bundle structure. Both $\theta^{i}_{P}(q)$ and $\theta^{i}_{Q}(q)$ are 1-forms on this submanifold, but the ambiguity in the form $\theta^{i}{}_{Q}$ is a right action of $H'$ (we can still permute the 1-forms), whereas the ambiguity in $\theta^{i}_{P}$ is $H'\times M/B$ (we can apply the group action independently at each point above $L=M/B$, in contrast to the first case). 

On the bundle $\tilde L(q)$ we now use the right $H'$ action to align all the $\theta^{i}_{P}$ with $\theta^{i}_{Q}$. Depending on the group $H'$, it is conceivable that this might not always be possible, but we will only consider cases where this is possible since otherwise it is clear that such a submersion cannot be legitimately called ``structure-preserving''. After this is done for $\tilde L(q)$ for all $q\in Q$, we have a situation that is quite similar to, but not the same as a reduction of principal bundle on $P'$. If we refer to \eqref{eq:70}, the subgroup consisting of elements $g_{ij}$ still acts on $P'$, but once this action is chosen on a point $x\in P'$, it is determined for all $y\in\tilde L(\tilde\sigma(x))$. The rest of the group elements do not have this restriction. Let us call this situation a \emph{semi-reduction} of the principal bundle.

An important simplification occurs where $g_{iA}=0$ in \eqref{eq:70}. This occurs frequently, as we shall see, since the group $H\subset GL(n)$ can have symmetries that force it to be zero due to the lower left entry being zero.

As we have hinted several times before, the connection components $\theta^{A}_{P}$ and $\omega^{AB}_{P}$ when restricted to a single leaf $\tilde L(q)$, form the Cartan connection on the leaf. For example, in the Riemannian case this reduces to equation \ref{eq:8}, where we just set the transversal basic forms to zero and all our techniques for dealing with immersions apply.

Let us return to \eqref{eq:73}. Now we have $\theta^{i}_{P}=\theta^{i}_{Q}$. First note that on $P'$, the form $\omega^{iA}_{P}$ are \emph{basic}, and if the special situation where $g_{iA}=0$ occurs, $\omega^{Ai}_{P}$ are basic as well. The best way to see this is to ponder the relationship between the Lie algebra and the group. Hence, we can write
\begin{equation}
  \label{eq:75}
  \omega^{Ai}_{P}=K^{Ai}{}_{B}\omega^{B}_{P}+M^{Ai}{}_{j}\omega^{j}{}_{Q}.
\end{equation}
Note that $K^{Ai}{}_{B}$ and $M^{Ai}{}_{j}$ are ordinary functions on the bundle $P'$. They can be considered tensor quantities on $M$, but since it is much easier to work with ordinary functions than tensors, we shall always do calculations on $P'$ and not on the base.
If $\omega^{Ai}{}_{P}$ are basic as well, there is another set which is symmetric with respect to \eqref{eq:75} in an obvious way, for example, orthogonal group would give $\omega^{Ai}_{P}=-\omega^{iA}{}_{P}$.

Considering from the point of view of giving the basic forms $\theta^{i}_{Q}$ and $\theta^{A}_{P}$ and then determining whether the distribution given by $\theta^{i}_{Q}$ really gives a distribution, Frobenius theorem tells us that this is the case if and only if
\begin{equation}
  \label{eq:76}
  d\theta^{i}_{Q}\equiv\theta^{A}_{P}\wedge\omega^{A}{}_{iP}\equiv0\mod\{\theta^{i}_{Q}\}.
\end{equation}
This is the \emph{integrability condition} that we need to enforce. If $\theta^{iA}_{P}$ and $\theta^{Ai}_{P}$ are linearly dependent, this will reduce to a condition on $K^{Ai}{}_{B}$.

Using $R^{*}_{h}(\omega)=\Ad(h^{-1})\omega$, we can verify that after the semi-reduction the 1-forms $\omega^{ij}_{P}$ also has no ambiguity left along each $\tilde L$. However, it is not necessarily true that $\omega^{ij}_{P}=\omega^{ij}_{Q}$. In fact, we can calculate $d\theta^{i}_{P}$ and $d\theta^{i}_{Q}$ on $P$ and $Q$ separately, and then use $\theta^{i}_{P}=\theta^{i}_{Q}$ to compare the results. We have
\begin{equation}
  \label{eq:77}
  \omega^{ij}_{P}=\omega^{ij}_{Q}-M^{A}{}_{ij}\theta^{A}_{Q}
\end{equation}
However, we can think of both $\omega_{P}^{ij}$ and $\omega_{Q}^{ij}$ as $\mathfrak{h'}$-valued 1-forms where $\mathfrak{h'}$ denotes the Lie algebra of $H'$, and hence this shows that $M^{A}{}_{ij}\theta^{A}_{Q}$ is also a $\mathfrak{h'}$-valued 1-form. This is the \emph{Lie algebra compatibility condition}.

These are all the conditions that we can deduce from our motivation. We can now finally put everything together and give our definition for structure-preserving submersion:

\begin{dfn}
  [Structure-preserving submersion] Let $B$ and $M$ be Cartan geometries with group $H'$ and $H$, respectively. A submersion $\sigma:M\rightarrow B$ is called \emph{structure-preserving} if after reduction of the principal bundle of $M$ adapted to the submersion to $P'$, we can find function $M^{A}{}_{ij}$ and $K^{Ai}{}_{B}$ on $P'$ such that, after right-$H_{L}$ action if needed, on the bundle we have
  \begin{itemize}
  \item $\theta^{i}_{P}=\theta^{i}_{Q}$,
  \item $\omega^{ij}_{P}=\omega^{ij}_{Q}-M^{A}{}_{ij}\theta^{A}_{Q}$ (this implies the Lie algebra-compatibility condition),
  \item $\omega^{Ai}_{P}=K^{Ai}{}_{B}\omega^{Q}_{P}+M^{Ai}{}_{j}\omega^{j}_{Q}$ (this implies the integrability condition, since we started with a submersion),
  \end{itemize}
\end{dfn}

Note that we specified the geometries on $B$ and $M$ independently. \emph{It is important to specify at least what class of geometry they each belong to}, for example, in the Riemannian case if a structure-preserving submersion exists then it can be interpreted uniquely in our above framework. However, if we allow $B$ to be non-torsion-free, then this is no longer the case: $M^{A}{}_{ij}$ can now take \emph{any} value subject to Lie-algebra compatibility, as any value different from \eqref{eq:77} would now contribute to torsion in the geometry $B$.

Although we started with the space $P$ and derived what a ``structure-preserving submersion'' should look like, making the space $Q$ look derivative, in our definition the centre stage is not given to $P$. This is deliberate: if we look at the definition of a structure-preserving immersion, i.e., ambient geometry in section \ref{sec:immers-subm}, we see that the centre stage is given to what the geometrical construction, i.e., the immersion, can ``see'' and we omitted altogether those parts unavailable to $Q$. Here we do the same thing, as this way of definition is more general. In addition, later in section \ref{sec:anoth-view-struct} we shall see that there are cases where $Q$ instead of $P$ is considered more immediately available to us.

Now, as usual, we will give examples of how this construction realises in the case of Riemannian and conformal geometries.

\begin{ex}
  [Riemannian submersion]

Now let us return to the torsion-free Cartan geometry modelled on Euclidean geometry. Let $f:M\rightarrow N$ be a submersion. Using essentially the same reduction of the principal bundle $P$ as we did before, we now have the Cartan connection on $P$ over $M$ of the form
\begin{equation}
  \label{eq:14}
  \omega=
  \begin{pmatrix}
    0&0&0\\
    \theta&\alpha&-\beta^{t}\\
    \psi&\beta&\gamma
  \end{pmatrix}
\end{equation}
where $\alpha$ is $n\times n$, $\gamma$ is $(n-m)\times(n-m)$, i.e., the forms $\theta$ constitute a basis of basic forms on $N$. As before, after reduction the forms $\beta$ becomes basic, hence
\begin{equation}
  \label{eq:15}
  \beta^{A}{}_{i}=M^{A}{}_{ij}\theta^{j}+K^{A}{}_{Bi}\psi^{B}.
\end{equation}

Now the submersion we have is very specific: it endows the space $N$ with a Riemannian structure. First, let us now use the right action to align all the forms $\theta$ so that they are basic for the projection $\pi: M\rightarrow N$. After this is done, we are no longer allowed to apply the right action of the group element
\begin{equation}
  \label{eq:16}
  h=
  \begin{pmatrix}
    1&0&0\\
    0&r&0\\
    0&0&1
  \end{pmatrix}
\end{equation}
arbitrarily on the whole space $M$: $h$ can still be applied, but once it is applied to a point in a leaf, the same action must be applied to all points on the same leaf. In other words, the degree of freedom is now in $N$, not $M$.

For the Riemannian submersion, the space $N$ also has a torsion-free Cartan connection
\begin{equation}
  \label{eq:17}
  \begin{pmatrix}
    0&0\\
    \theta&\bar\alpha
  \end{pmatrix}
\end{equation}
notice we have used the same $\theta$ as in the space $M$. To be more precise, the $\theta$ on $M$ is now the pullback of the $\theta$ on $N$, but for brevity we omit the pullback signs. This should not cause any confusion. Calculating $d\theta$ in two ways, first on $N$, then on $M$, we obtain
\begin{equation}
  \label{eq:18}
  d\theta=\bar\alpha\wedge\theta=\alpha\wedge\theta-\beta^{t}\wedge\psi.
\end{equation}
Let us see what kind of restriction this places on our quantities. Expanding using \eqref{eq:15}, we have (recall: exterior derivatives commute with pullbacks)
\begin{align}
  \label{eq:19}
  \bar\alpha^{i}{}_{j}\wedge\theta^{j}&=\alpha^{i}{}_{j}\wedge\theta^{j}-M^{A}{}_{ij}\theta^{j}\wedge\psi^{A}-K^{A}{}_{Bi}\psi^{B}\wedge\psi^{A}\\
&=(\alpha^{i}{}_{j}+M^{A}{}_{ij}\psi^{A})\wedge\theta^{j}-K^{A}{}_{Bi}\psi^{B}\wedge\psi^{A}
\end{align}
comparing coefficients, we learn that $\alpha^{i}{}_{j}=\bar\alpha^{i}{}_{j}-M^{A}{}_{ij}\psi^{A}$, $M^{A}{}_{ij}$ is skew symmetric in $i$ and $j$ (due to the skew-symmetry of $\bar\alpha$), and $K^{A}{}_{Bi}$ is symmetric in $A$ and $B$. We remark that the redefinition of $\alpha$ by $\bar\alpha$ is the same technique used in the equivalence method for absorbing torsion. In fact, if we do not do this, the quantities $M^{A}{}_{ij}$ really corresponds to geometrical torsion of the space $B$.

Putting everything together, our Cartan connection is now
\begin{equation}
  \label{eq:22}
  \omega=
  \begin{pmatrix}
    0&0&0\\
    \theta&\bar\alpha-M\psi&-\theta^{t}M^{t}-\psi^{t}K^{t}\\
    \psi&M\theta+K\psi&\gamma
  \end{pmatrix}
\end{equation}
It should be noted that \eqref{eq:22} is a vast simplification from \eqref{eq:14} because only $\gamma$ and $\psi$ are not pullbacks of forms defined on $N$, and the functions $M$ and $K$ have quite strong symmetries. Let us also calculate the curvature
\begin{equation}
  \label{eq:23}
  \Omega=
  \begin{pmatrix}
    0&0&0\\
    (2,1)&(2,2)&\star\\
    (3,1)&(3,2)&(3,3)
  \end{pmatrix}
\end{equation}
where the independent entries are
\begin{align}
  \label{eq:24}
  (2,1)&=d\theta+\bar\alpha\wedge\theta-M\psi\wedge\theta-\theta^{t}M^{t}\wedge\psi(-K^{t}\psi^{t}\wedge\psi)=0\\
  \label{eq:x3.1}
  (3,1)&=d\psi+M\theta\wedge\theta+K\psi\wedge\theta+\gamma\wedge\psi=0\\\label{eq:x2.2}
  (2,2)&=d\bar\alpha-dM\wedge\psi-Md\psi+\bar\alpha\wedge\bar\alpha-\theta^{t}M^{t}\wedge M\theta\\
  &\quad-\psi^{t}K^{t}\wedge M\theta-\theta^{t}M^{t}\wedge K\psi+M\psi\wedge M\psi-\psi^{t}K^{t}\wedge K\psi\\
\label{eq:x3.3}
  (3,3)&=d\gamma+\gamma\wedge\gamma-M\theta\wedge M\theta-K\psi\wedge\theta^{t}M^{t}-M\theta\wedge\psi^{t}K^{t}-K\psi\wedge\psi^{t}K^{t}\\
  (3,2)&=dM\wedge\theta+Md\theta+dK\wedge\psi+Kd\psi\\
  &\quad +M\theta\wedge\bar\alpha+\gamma\wedge M\theta+K\psi\wedge\alpha+\gamma\wedge K\psi-M\theta\wedge M\psi-K\psi\wedge M\psi
\end{align}
the first two are set to zero because of the torsion-free requirement.
\end{ex}
\begin{ex}
  [Conformal submersions]
Now let us investigate the submersion problem. Again we first align the first $\theta$ to be tangential to the submersion, after which we have the Cartan connection
\begin{equation}
  \label{eq:62}
    \omega=
  \begin{pmatrix}
    \epsilon&\nu&\mu&0\\
    \theta&\alpha&-\beta^{t}&\nu^{t}\\
    \psi&\beta&\gamma&\mu^{t}\\
    0&\theta^{t}&\psi&-\epsilon
  \end{pmatrix}
\end{equation}
for the conformal submersion, the space $N$ also has a normal M\"obius geometry defined on it:
\begin{equation}
  \label{eq:63}
  \begin{pmatrix}
    \epsilon&\bar\nu&0\\
    \theta&\bar\alpha&\bar\nu^{t}\\
    0&\theta^{t}&-\epsilon
  \end{pmatrix}
\end{equation}
again we have aligned the $\theta$ on $M$ and on $N$. Note that using the group elements $p$, the $\epsilon$ on both space have also been made the same. This requires that the group degrees of freedom of the elements
\begin{equation}
  \label{eq:64}
  \begin{pmatrix}
    1&p&0&s\\
    0&1&0&p^{t}\\
    0&0&1&0\\
    0&0&0&1
  \end{pmatrix}
  \begin{pmatrix}
    z&0&0&0\\
    0&a&0&0\\
    0&0&1&0\\
    0&0&0&z^{-1}
  \end{pmatrix}
\end{equation}
are now only available on $N$, not on $M$ anymore. Once more,
\begin{equation}
  \label{eq:65}
  \bar\alpha\wedge\theta=\alpha\wedge\theta-\beta^{t}\wedge\psi
\end{equation}
writing
\begin{equation}
  \label{eq:66}
  \beta^{A}{}_{i}=M^{A}{}_{ij}\theta^{j}+K^{A}{}_{Bi}\psi^{B}
\end{equation}
we have the same restriction as in the Riemannian case: $M^{A}{}_{ij}$ is skew in $i$ and $j$, whereas $K^{A}{}_{Bi}$ is symmetric in $A$ and $B$.
In the conformal case, there is actually a hidden condition in this: consider the adjoint $H$ action by the element
\begin{equation}
  \label{eq:67}
  \begin{pmatrix}
    1&0&q&s\\
    0&1&0&0\\
    0&0&1&q^{t}\\
    0&0&0&1
  \end{pmatrix}
\end{equation}
the degree of freedom of which is still on the whole of $M$. We have
\begin{equation}
  \label{eq:68}
  \beta\mapsto \beta-\theta^{t}q,\qquad M^{A}{}_{ij}\mapsto M^{A}{}_{ij}-\delta_{ij}q^{A}
\end{equation}
hence, to maintain the antisymmetry of $M^{A}{}_{ij}$, we must set $q=0$.
Let us recap here: the degree of freedom on $N$ is \eqref{eq:64}, whereas the degree of freedom on $M$ is only
\begin{equation}
  \label{eq:69}
  \begin{pmatrix}
    1&0&0&0\\
    0&1&0&0\\
    0&0&c&0\\
    0&0&0&1
  \end{pmatrix}
\end{equation}
i.e., only rotation of the $\psi$ is allowed.

Now, exterior differentiate the forms, we can obtain similar constraints to those obtained in the Riemannian case. As we will not use them in subsequent discussions, we omit them here.
\end{ex}

\subsection{The inapplicability of the equivalence method}
\label{sec:equiv-meth-fails}
Suppose that we are explicitly given two structure-preserving submersions $M\rightarrow N$ and $M'\rightarrow N'$, and we ask the question that if they are equivalent. This can be solved by Cartan's equivalence method: first, we can confirm by calculation that both of the submersions are really structure-preserving. Then, using the method of equivalence, we can check that these two are equivalent as submersions. Then, by uniqueness, these two are equivalent as structure-preserving submersions.

This is all very good, but if the problem can be solved straightforwardly using the equivalence method, then why do we need to develop our method using a mixture of Cartan geometry and lifting, which now seems to complicate things? The answer lies in the fact that the above algorithm can only be applied when the structure-preserving submersions are giving explicitly, by giving expressions involving local coordinates, for example. Suppose that we want to prove some general properties of a certain class of structure-preserving submersions, for example, classification. Now if we apply the equivalence method for general submersions, as we do not have explicit expressions for the coframe there is in general no way that we can ensure or check that the submersion is structure-preserving under the framework of the method of equivalence.

The problem is that, as we have shown, though in the structure-preserving case we still have a coframe problem, the relations between the old and new coframes are no longer expressed by a simple Lie group transformation. Indeed, for any two coframes at two base points, we can still find an element of the Lie group that transforms the old coframe into the new one, but now this transformation depends on the base point in a non-trivial and non-local way: if we specify the transformation at one base point, \emph{parts, but not all,} of the transformation at some other points are determined, yet at another set of points they remain completely undetermined.

It is reasonable to ask if a straightforward extension to the method of equivalence could be used for the structure-preserving submersion problems. After all, we still have group actions, and an absorption procedure is still possible, if we separate the parts that can be applied to all points and the parts that can only be independently applied to a subset. Indeed this is possible, but the main problem is that this does now get us very far: the usual formulation of the method of equivalence has few, if any, constraints on the torsion elements. However, as we will see in section \ref{sec:riem-rigid-flow}, structure-preserving submersion questions are characterised by a large number of interconnected constraints on $M^{A}{}_{ij}$ and $K^{Ai}{}_{B}$. Thus the problem lies not so much in finding the correct ``frame'', but to get meaning from the messy constraints. In simple cases, for example, in the case where the codimension of the submersion is 1, as in \ref{sec:riem-rigid-flow}, it is possible to get results by elementary means. However, in yet more complicated cases a more efficient method for dealing with the constraints needs to be developed.

\subsection{Some examples}
\label{sec:some-examples}
Here we give several examples of how real, specific problems arising in physics can be fitted to our framework.
\begin{ex}
  [Pseudo-Riemannian submersion of codimension 2 in 4-dimensional spacetime] In using structure-preserving submersions, it is not necessary to start giving the Cartan geometry on the total space. Instead, we can specify some extra constraints and deduce certain things about the total space. Let us try giving the following conditions:
  \begin{enumerate}
  \item $M\rightarrow B$ is a structure-preserving submersion.
  \item $H=SO(3,1)$, $H'=SO(1,1)$.
  \item The Cartan connection on each $\tilde L$ satisfies the structural equation of the 2-sphere.
  \item The Cartan connection on $P$ lies in the kernel of the Ricci homomorphism.
  \end{enumerate}
Item 3 above just says that each leaf has the full rotational symmetry, whereas item 4 says that the total space is Einstein. It is shown in \cite{Alvarez:2007p4209} that the unique solution satisfying these conditions is the Schwarzschild solution for black holes. It is desirable to use our method to study the existence of other black holes that can be described as a structure-preserving submersion.
\end{ex}

\begin{ex}
  [String fluid]
Another problem that can be put into our framework in a straightforward fashion is the problem of string fluid. This arises, for example, in physics where the effective action for unstable D-branes reduces to the case where relativistic string fluid of moving electric flux lines, see \cite{Gibbons:2000p8621}. Using \eqref{eq:14} with a few change of signs we have a working model for relativity, as we have already used for the above example. The string itself is 1-dimensional, so its trajectory in spacetime (the leaves) has group $SO(1,1)$, corresponding to $\gamma$ in \eqref{eq:14}. $\bar\alpha$ is then the curvature of spacetime the string fluid ``sees''. We see that in our framework it is now possible to give meaning to ``a string fluid moving rigidly in spacetime'', and using \eqref{eq:24} etc., it is now possible to calculate properties.
\end{ex}

\begin{ex}
  [Rigid motion in Newtonian spacetime] We consider a 3-dimensional body moving in Newtonian spacetime. The relevant group is now the Galiean group acting on spacetime:
\begin{equation}
\label{lie1}
\begin{pmatrix}
  1\\
  t'\\
  \mathbf{x}'
\end{pmatrix}
=
\begin{pmatrix}
 1&0&\mathbf{0}^{T}\\
 t_{0}&1&\mathbf{0}^{T}\\
 \mathbf{x}_{0}&\mathbf{v}&R
\end{pmatrix}
\begin{pmatrix}
  1\\
  t\\
  \mathbf{x}
\end{pmatrix}
\end{equation}
where $R\in SO(3)$. The Cartan connection is,
\begin{equation}
  \label{eq:80}
  \omega=
  \begin{pmatrix}
    0&0&0\\
    \tau&0&0\\
    \xi&\nu&\rho
  \end{pmatrix}
\end{equation}
where $\rho\in\mathfrak{so}(3)$. Note that this is without doing the first reduction. Note also that, since in Newtonian mechanics the form $\tau$ is \emph{always} aligned with motion and this alignment is preserved by the Galilean group, \emph{this is also the form of the 1-forms after reduction}, i.e., the first reduction does nothing at all! Hence, in contrast to the previous two cases, in Newtonian spacetime rigidity places no constraints on the system whatsoever.
\end{ex}
We see in \eqref{eq:40} that in general it is impossible to construct a covering map directly from the bundle of $M$ to the bundle of $B$, but as the above example shows, in Newtonian spacetime, such a map is trivially available to us. The following generalisation of this observation is immediate.

\begin{prop}
  A structure-preserving submersion places no extra constraints on the system (i.e., $M^{A}{}_{ij}$ and $K^{Ai}{}_{B}$ vanish identically) if and only if in the original Lie algebra the block $\omega^{iA}$ vanish identically.
\end{prop}
\begin{proof}
  The ``only if'' part is clear. For the ``if'' part, if this block in the original Lie algebra does not vanish, we can construct geometries where the Cartan connection corresponding to this part is non-zero. After the first reduction, there are also frames where this part, now consisting of basic forms only, remains non-zero. These equations yield the constraints.
\end{proof}

\subsection{Another view of structure-preserving submersions}
\label{sec:anoth-view-struct}
In all the examples we see above, the space $M$ is considered more immediately available to us. The space $B$ often seem no more than a mathematical construct, as even when $M$ is flat, $B$ needs not be. Furthermore, in all applications except that of the black hole, the space $M$ is given to us together with its geometry. This needs not be the case: there are examples where $B$ is immediately available to us whereas $M$ is somehow hidden from view:
\begin{ex}
  [Kaluza-Klein] Let $M=\mathbb{M}\times S^{1}_{r}$, where $\mathbb{M}$ is 4-dimensional Minkowski spacetime and $S^{1}_{r}$ is the circle with radius $r$, each given the usual metric. The space $M$ is given the product metric. Then the projection $\pi:M\rightarrow\mathbb{M}$ is a structure-preserving submersion. This submersion corresponds to the usual dimensional reduction in Kaluza-Klein theory.
\end{ex}
By itself this example is not very interesting: after all locally $M$ is the same as 5-dimensional Minkowski spacetime, which we already studied before. However, if we go to the closely-related concept of gauge theory, especially Yang-Mills theory, we see new things popping up. For example, in \cite{Watson:1983} which studies Yang-Mills theory on $S^{4}$, the following diagram is relevant:
\begin{equation}
  \label{eq:72}
    \xymatrix{
&SU(2)\ar[r]^{\pi_{4}}\ar[d]&P_{1}(\cxs)\ar[d]\\
U(1)\ar[r]&S^{7}\ar[r]^{\pi_{2}}\ar[d]_{\pi_{1}}&P_{3}(\cxs)\ar[dl]^{\pi_{3}}\\
&S^{4}&
}
\end{equation}
where all $\pi_{i}$ are Riemannian submersions provided that each space in the diagram is equipped with the standard metric. However, this approach is unsatisfactory since it relies on the introduction of Riemannian metrics on the various spaces, while Yang-Mills theory itself is built from group-theoretic pieces. Hence it is desirable to study such problems using our general, group-theoretic structure-preserving submersions alone. This programme will be pursued in a subsequent paper.

\section{In-depth example: rigid Riemannian flow}
\label{sec:riemannian-geometry}

The case of codimension-1 Riemannian submersion is interesting in that under some additional mild conditions the geometry of the whole flow is rigidly determined. We shall now begin to derive these seemingly surprising results. For background information, see \cite{me:000,born1909,herglotz1910,noether1910,pirani1962,rayner1959,Wahlquist:1967p2556,Wahlquist:1966p2548,Estabrook:1964p2608,Giulini:2006p66}.

Up until now we have avoided doing massive computations. Since this is unavoidable now, we need a clear way to express our quantities instead of carrying exterior derivatives of forms around all the time. Here it is convenient to use covariant derivatives in the principal bundle. Covariant derivatives on the bundle can only be defined in cases where the vector space $\mathfrak g/\mathfrak h$ which is isomorphic to the tangent space at each point is invariant under the $\Ad(H)$ representation when considered as a Lie-module, meaning that the interpretation of the usual ``horizontal subspaces'' makes invariant sense. Since any tensor on the base manifold is interpreted as a vector-valued \emph{function} on the bundle where the vector space is a representation of the group $H$, given a ``direction'' $\mathbf{v}\in\mathfrak g/\mathfrak h$ the inverse of the Cartan connection $\omega$ maps this direction to a vector field on $P$: $\mathbf{v}^{\dagger}=\omega^{-1}(\mathbf{v})$. Then given a tensor, which we will write as
\begin{equation}
  \label{eq:30}
  T=T^{A}\otimes\mathbf{e}_{A}
\end{equation}
where $\mathbf{e}_{A}$ is a basis for the vector space $T$ is taking value in, the covariant derivative is then
\begin{equation}
  \label{eq:31}
  \nabla_{\mathbf{v}}T=\omega^{-1}(\mathbf{v})(T^{A})\otimes\mathbf{e}_{A}\equiv T^{A}{}_{;\mathbf{v}}\otimes\mathbf{e}_{A}
\end{equation}
a trick of calculating the covariant derivative of a function is to just calculate the exterior derivative of the scalar part to get a 1-form, and the terms involving the horizontal 1-forms are the covariant derivatives, e.g.:
\begin{equation}
  \label{eq:36}
  dT^{A}|_{\theta^{1}}\equiv \text{the terms in $dT^A$ involving $\theta^{1}$}=T^{A}{}_{;1}\theta^{1}.
\end{equation}

\subsection{Riemannian rigid flow in homogeneous space}
\label{sec:riem-rigid-flow}

This is the case where $\dim M-\dim N = 1$, which gives us vast simplification: now $\gamma = 0$ and (let us denote by  $0$ the only index taken by $A,B,\dots$)
\begin{equation}
  \label{eq:21}
  M^{0}{}_{ij}\equiv M_{ij},\qquad K^{0}{}_{0i}\equiv K_{i}.
\end{equation}

\begin{prop}
  For Riemannian rigid flow, if the ambient space is homogeneous, then $M$ is basic for the projection $\pi:M\rightarrow N$.
\end{prop}
\begin{proof}
  Using \eqref{eq:x2.2} and \eqref{eq:x3.1}, we can collect the terms in \eqref{eq:x2.2} that are of the basis $\theta\wedge\theta$, i.e.,
  \begin{equation}
    \label{eq:25}
    (2,2)|_{\theta\wedge\theta}=d\bar\alpha+\bar\alpha\wedge\bar\alpha-\theta^{t}M^{t}\wedge M\theta-MM\theta\wedge\theta.
  \end{equation}
Now the left hand side is constant since $M$ is homogeneous, $d\bar\alpha+\bar\alpha\wedge\bar\alpha$ is the curvature of the space $N$ and hence is basic. This shows the expression
\begin{equation}
  \label{eq:26}
  \theta^{t}M^{t}\wedge M\theta+MM\theta\wedge\theta
\end{equation}
must be basic as well. Using indices, this means that the function
\begin{equation}
  \label{eq:27}
  R_{ijkl}=M_{ik}M_{jl}-M_{il}M_{jk}+2M_{ij}M_{lk}
\end{equation}
is constant along each leaf. Now, for example
\begin{equation}
  \label{eq:34}
  R_{1212}=-M_{12}M_{21}+2M_{12}M_{21}=-M_{12}^{2}
\end{equation}
since the function $M$ is smooth, this shows that $M_{12}$ is constant along each leaf. The claim is proved by noting that we can choose indices to generate $-M_{ij}^{2}=\text{const.}$ along each leaf for any pair of $(i,j)$.
\end{proof}
Let us remark again that it is not true that we can always use the above algebraic method to show that a function is basic for a projection: usually the right $H$ action will change the value of the function along the leaf. Our method makes sense in this case due to the fact that by using the pulled back form $\theta$ on the principal bundle over $M$, we have already killed the degree of freedom that affects $M_{ij}$, i.e., in the reduced bundle, the function $M_{ij}$ is constant in each fibre, as can be shown by explicit calculation.

Also, for homogeneous spaces we can set $(2,2)=0$ straight away by a process called mutation: instead of considering the Euclidean model \eqref{eq:1}, we can use a spherical model
\begin{equation}
  \label{eq:32}
    \mathfrak g'=\left\{
    \begin{pmatrix}
      0&-v\\
      v&a
    \end{pmatrix}\in M_{n+1}(\rs)\mid a\in \mathfrak o(n)
\right\}
\end{equation}
or a hyperbolic model
\begin{equation}
  \label{eq:33}
    \mathfrak g''=\left\{
    \begin{pmatrix}
      0&v\\
      v&a
    \end{pmatrix}\in M_{n+1}(\rs)\mid a\in \mathfrak o(n)
\right\}  
\end{equation}
such mutation of model does not affect the geometry in anyway, but globally subtracts from any curvature function a constant (in the spherical model, a sphere of unit diameter is ``flat'' whereas a plane is not). By setting the scale correctly, we can comfortably set all curvature blocks in $M$ to zero.
\begin{cor}
\label{kmmagic}
  For the assumption of the above proposition, $K_{[i;j]}=0$.
\end{cor}
\begin{proof}
  For \eqref{eq:x3.1}, we exterior differentiate again. Note that we have a term $K_{[i;j]}\psi^{A}\wedge\theta^{i}\wedge\theta^{j}$. Hence we need to collect the other terms that have the basis $\theta\wedge\theta\wedge\psi$. It is easy to show that the only contribution is proportional to $MM\theta\wedge\theta\wedge\psi$. However, looking at \eqref{eq:x3.3}, it is clear that this term vanishes.
\end{proof}
\begin{cor}
\label{basick}
  For the assumption of the above proposition, if in addition $M\neq 0$, then $K$ is basic.
\end{cor}
\begin{proof}
  We again use \eqref{eq:x2.2} and \eqref{eq:x3.1}, but now focus on terms containing $\theta\wedge\psi$. The only term that needs discussion is $-dM\wedge\psi$, but since $M$ is basic, $dM$ is as well. Hence the expression
  \begin{equation}
    \label{eq:35}
    B\wedge\psi = -MK\theta\wedge\psi-\psi^{t}K^{t}\wedge M\theta-\theta^{t}M^{t}\wedge K\psi
  \end{equation}
where $B$ is some basic form for the projection $\pi:M\rightarrow N$. Writing out using coordinates, this simplifies to the fact that 
$K_{i}M_{jk}$ is basic. Since some component of $M$ does not vanish, it follows that $K$ must also be basic.
\end{proof}
\begin{thm}
\label{herglotz1}
  If the ambient space is homogeneous and $M\neq 0$, then a Riemannian rigid flow is isometric.
\end{thm}
\begin{proof}
  The easiest way to show this is to take a section of the bundle $P$ and calculate the metric. The Riemannian metric is now
  \begin{equation}
    \label{eq:37}
    g=\psi^{0}\otimes\psi^{0}+\sum_{i}\theta^{i}\otimes\theta^{i}.
  \end{equation}
Let $\mathbf{I}_{0}$ be dual to $\psi^{0}$, $\mathbf{I}_{i}$ be dual to $\theta^{i}$. For a vector field $\mathbf{V}=\lambda\mathbf{I_{0}}$ where $\lambda$ is a function on $M$, we have
\begin{equation}
  \label{eq:38}
  \mathcal{L}_{\mathbf{V}}g=\sum_{ij}\lambda M_{ij}(\theta^{i}\otimes\theta^{j}+\theta^{j}\otimes\theta^{i})+\mathbf{I}_{0}(\lambda)\psi^{0}\otimes\psi^{0}+(\mathbf{I}_{i}(\lambda)-\lambda K_{i})(\theta^{i}\otimes\psi^{0}+\psi^{0}\otimes\theta^{i}).
\end{equation}
If for some $\lambda$ the whole thing vanishes, we have an isometry. The term involving $M_{ij}$ vanishes automatically dual to antisymmetry. The question now reduces to: can we find $\lambda$ such that
\begin{equation}
  \label{eq:39}
  \mathbf{I}_{0}(\lambda)=0,\qquad \mathbf{I}_{i}(\lambda)=\lambda K_{i}
\end{equation}
The first condition just says that $\lambda$ is required to be basic for the projection $\pi:M\rightarrow N$ as well.
Hence we will look for a solution on $N$. The second condition can be rewritten as
\begin{equation}
  \label{eq:41}
  K_{i}=\frac{\pd \log\lambda}{\pd x^{i}}
\end{equation}
for some coordinate system $x^{i}$ on $N$. Poincar\'e{} Lemma says that this can be locally solved as long as $K_{[i;j]}=0$, but this we already know. Finally, since $K_{i}$ is basic, $\lambda$ constructed this way is also basic.
\end{proof}

\subsection{Riemannian rigid flow in conformally flat space}
\label{sec:conf-flat-case}
While we are at it, let us prove the following conceptually simple but calculationally very intense result
\begin{thm}
  If the ambient space is conformally flat and $M\neq 0$, then a Riemannian rigid flow is isometric.
\end{thm}
\begin{proof}
  By the proof of theorem \ref{herglotz1}, we need to show that $K_{i}$ in this case is basic and $K_{[i;j]}=0$. Corollary \ref{kmmagic} does not use properties of the ambient space, hence the latter condition reduces to showing that $M_{ij}$ is basic. Let us first state our strategy:
  \begin{enumerate}
  \item eliminate all terms containing $K$ and derivatives of $M$ in equations, leaving only products of $M$ and other functions;
  \item eliminate all non-basic functions by using the conformal flatness condition;
  \item the remaining equations then give algebraic constraints on $M$ in terms of basic functions.
  \end{enumerate}
To prepare for our calculation, let us write \eqref{eq:24} onwards using a basis for the Lie algebra, and use covariant derivatives. We have the following conditions
\begin{align}
  R_{0i0j}&=-K_{(i;j)}+K_{i}K_{j}+M_{kj}M^{k}{}_{i}\label{1stweyl}\\
  R_{0ijk}&=-M_{ik;j}+M_{ij;k}-2K_{i}M_{kj}\\
  R_{ij0k}&=-M_{ij;k}-K_{i}M_{jk}+K_{j}M_{ik}+K_{k}M_{ij}\\
  R_{ijkl}&=\widetilde R_{ijkl}+M_{ik}M_{jl}-M_{il}M_{jk}-2M_{ij}M_{lk}\\
  R_{00}&=-K^{i}{}_{;i}+K_{i}K^{i}+M^{ij}M_{ij}\\
  R_{0i}&=M^{j}{}_{i;j}+2K^{j}M_{ij}\\
  R_{ij}&=\widetilde R_{ij}+2M_{ik}M_{j}{}^{k}-K_{i}K_{j}+\tfrac{1}{2}(K_{i;j}+K_{j;i})\\
  R&=\widetilde R+2K^{i}{}_{;i}-2K_{i}K^{i}+M_{ij}M^{ij}
\end{align}
where we have repeatedly used theorem \ref{kmmagic}. Here $R$ denotes the Riemannian curvature \emph{function} on $M$ and the associated Ricci curvature functions, whereas $\tilde R$ denotes the same thing on $N$. Before, we used the homogeneous condition to set all $R$ to be constants and used algebraic conditions hidden in the set of equations to force the components of $M$ to be constant. Now the problem is more difficult, since $R$ could vary from point to point even along each leaf. We can only be sure that the Weyl part of the curvature function vanishes. The Weyl curvature function is defined in terms of the Riemann curvature function as
\begin{align}
  \label{eq:100}
  W_{\mu\nu\rho\lambda}&=R_{\mu\nu\rho\lambda}-\frac{2}{n-2}(\delta_{\mu[\rho}R_{\lambda]\nu}-\delta_{\nu[\rho}R_{\lambda]\mu})+\frac{2}{(n-1)(n-2)}R\delta_{\mu[\rho}\delta_{\lambda]\nu}\\
&=R_{\mu\nu\rho\lambda}-\delta_{\mu\rho}F_{\lambda\nu}+\delta_{\mu\lambda}F_{\rho\nu}+\delta_{\nu\rho}F_{\lambda\mu}-\delta_{\nu\lambda}F_{\mu\rho}
\end{align}
where in the last line
\begin{equation}
  \label{eq:20}
  F_{\lambda\nu}=\frac{1}{n-1}R_{\lambda\nu}+\frac{1}{2(n-1)(n-2)}R\delta_{\lambda\nu}
\end{equation}
Using the above, let us calculate the corresponding Weyl equation to \eqref{1stweyl}:
\begin{align}
  W_{0i0j}&=\left(\frac{1}{n-2}\widetilde R_{ij}-\frac{1}{(n-1)(n-2)}\widetilde R\delta_{ij}\right)\label{aline}\\
&\quad+\left(M_{kj}M^{k}{}_{i}+\frac{2}{n-2}M_{ik}M_{j}{}^{k}-\frac{1}{(n-1)(n-2)}M_{kl}M^{kl}\right)\label{bline}\\
&\quad+\left(\frac{n-3}{n-2}(K_{i}K_{j}-K_{(i;j)})+\frac{2}{(n-1)(n-2)}\delta_{ij}(K_{k}K^{k}-K^{k}{}_{;k})\right)\label{cline}
\end{align}
Observe: the left hand side and the first line are basic functions for the projection to $N$, the second line is a quadratic function of the components of $M_{ij}$, and the third line is exactly the terms containing $K$ and its derivatives occurring in $F_{ij}$ in \eqref{eq:20}. This means that by substitution we can exchange $F_{\lambda\nu}$ for a sum of basic functions and quadratic products of the components of $M_{ij}$. For $W_{ijkl}$, we have
\begin{equation}
  \label{eq:28}
    W_{ijkl}=R_{ijkl}-\delta_{ik}F_{lj}+\delta_{il}F_{kj}+\delta_{jk}F_{li}-\delta_{jl}F_{ik}
\end{equation}
now $R_{ijkl}$ is a sum of basic functions and quadratic products of $M_{ij}$ as well. So schematically we have
\begin{equation}
  \label{eq:2}
  0=(\text{quadratic terms of $M_{ij}$})+(\text{basic functions})
\end{equation}
expanding, we get
\begin{equation}
  \label{eq:29}
    M_{ik}M_{jl}-M_{il}M_{jk}-2M_{ij}M_{lk}=\text{basic function}
\end{equation}
hence $M_{ij}$ is basic as well. As all terms involving $M$ are now basic, the problem simplifies greatly and the calculation analogous to corollary \ref{basick} shows that $K$ is basic as well as long as $M\neq 0$.
\end{proof}

Even though we have not yet started investigating conformal geometry, we already have the following result:
\begin{cor}
  For a flat conformal geometry, a conformally rigid flow is automatically conformally isometric as long as $M\neq0$ in some representative Riemannian structure.
\end{cor}
\begin{proof}
  Since a conformally structure allows scaling by non-zero functions, the condition $M\neq 0$ is conformally invariant. Since the flow is conformally rigid, it preserves the metric up to scale, and by using the scaling we can find a Riemannian representative that is conformally flat and in which the flow is a Riemannian rigid flow. Hence in this representative the flow is an isometry, and in the conformal geometry the original flow must be conformally isometric.
\end{proof}

\subsection{The Ricci-flat case}
\label{sec:ricci-flat-case}

The relevant equations in this case are:
\begin{align}
  \label{eq:74}
    R_{00}&=0=-K^{i}{}_{;i}+K_{i}K^{i}+M^{ij}M_{ij}\\
  R_{0i}&=0=M^{j}{}_{i;j}+2K^{j}M_{ij}\\
  R_{ij}&=0=\widetilde R_{ij}+2M_{ik}M_{j}{}^{k}-K_{i}K_{j}+\tfrac{1}{2}(K_{i;j}+K_{j;i})\\
  R&=0=\widetilde R+2K^{i}{}_{;i}-2K_{i}K^{i}+M_{ij}M^{ij}
\end{align}
Using the first and last equations we get 
\begin{equation}
  \label{eq:81}
 \sum_{ij}M^{ij}M_{ij}=\sum_{i}(K^{i}{}_{;i}-K^{i}{}_{i})=\text{constant along each leaf} 
\end{equation}
 Now there are not enough indices for doing the kind of computations we did before and we can only conclude  that the sum of the square of all terms is constant, whereas $M_{ij}$ itself can undergo rotation in the leaf direction. However, there are still other constraints to satisfy, and this still places serious constraints on the existence of rigid flow.

\subsection{The uniqueness of rigid flow}

Even in general curves background, when the Herglotz-Noether theorem and its generalisations are false (there might not be any Killing vectors at all), we still have the following uniqueness result concerning rotational rigid flow.
\label{sec:exist-uniq-flow}

\begin{thm}
  Assume that we have a rigid flow in an arbitrary Riemannian manifold $M$ and let $T\subset M$ be a hypersurface transversal to the flow. If the value of $M_{ij}$ is known on $T$, then it is known throughout $M$. In addition, for any open set containing $T$ in which $M_{ij}$ does not vanish, the value of $K_{i}$ throughout the open set is determined by the value of $M_{ij}$ value on $T$.
\end{thm}
\begin{proof}
  Look at \eqref{1stweyl} and below. For the equation involving $R_{ijkl}$, giving $M_{ij}$ at one point determines $\tilde R_{ijkl}$ for all points along the leaf, and together with values of $R_{ijkl}$ along the leaf, all values of $M_{ij}$ along the same leaf are determined. If in an open set $M_{ij}\neq 0$, then the equation for $R_{0i}$ will allow us to solve for $K^{i}$ uniquely given $M_{ij}$.
\end{proof}
Of course, there is no reason why any rigid flow should exist at all in a general space. An example of a rotational rigid flow in a non-conformally-flat spacetime is given in \cite{pirani1964}.

On the other hand, in a flat space if the flow is non-rotating but isometric, the above theorem is false, as one can easily give many examples by ``dragging the flow along''. See \cite{Giulini:2006p66} for an explicit construction.

\section{Conclusion}
\label{sec:conclusion}
In this paper we have described a general framework for dealing with problems of a structure-preserving submersion between manifolds. We gave some examples of how real problems can be adapted to our framework, and by using our framework we successfully extended the classical Herglotz-Noether theorem to all conformally flat spaces in all dimensions. Several interesting projects that fit within our framework that could be taken in the future are:
\begin{itemize}
\item Study the structure-preserving submersions arising in Yang-Mills and other gauge theories without introduction of any metric, using group properties only.
\item Study black hole solutions in higher dimensions that has constraints that can be given by structure-preserving submersions.
\item Study structure-preserving submersions arising in the fluid description of string theory, in the conformal framework \cite{Bhattacharyya:2007p31}.
\item Explore the adaptation of the framework to the supersymmetric case.
\end{itemize}
\bibliography{submersion}{}
\bibliographystyle{hplain}
\addcontentsline{toc}{section}{References}

\end{document}